\documentclass{quantumdust}

\usepackage[T1]{fontenc}
\usepackage{amssymb}
\usepackage{mathrsfs}
\usepackage{stmaryrd}
\usepackage{tikz-cd}
\usepackage{booktabs}
\usepackage{pgfplots}
\pgfplotsset{compat=1.18}
\pgfplotsset{qd/.style={
  tick label style={font=\scriptsize}, label style={font=\footnotesize},
  title style={font=\footnotesize}, legend cell align=left,
  legend style={font=\scriptsize, draw=black!60, fill=white, fill opacity=0.85,
                text opacity=1, inner sep=1pt},
  grid=both, grid style={black!10}, axis line style={black!80},
  every axis plot/.append style={line width=0.7pt}}}

\makeatletter
\def\endfront@text{}
\def\enddoc@text{}
\def\@setamsclass{}
\def\@setkeywords{}
\def\@settitlecomment{}

\makeatother

\newtheorem{theorem}{Theorem}[section]
\newtheorem{lemma}[theorem]{Lemma}
\newtheorem{proposition}[theorem]{Proposition}
\newtheorem{corollary}[theorem]{Corollary}
\theoremstyle{definition}

\theoremstyle{remark}

\begin{document}

\title[Quantum Dust from the Curse of Dimensionality]{Quantum Dust from the Curse of Dimensionality}
\author{Kenan Oggad}
\address{Universit\'{e} Paris-Saclay}
\email{kenan.oggad@universite-paris-saclay.fr}
\keywords{spectral dimension, concentration of measure, Fubini--Study geometry,
  quantum gravity, heat kernel, regular variation}
\subjclass[2020]{81P45, 83C45, 58J35, 60B20}

\begin{abstract}
Why do unrelated approaches to quantum gravity nearly all find spacetime two-dimensional at the shortest scales? Each answers with its own dynamics; we exhibit a single kinematic route to the same two, assuming no field equation. Concentration of measure on the Fubini--Study geometry of pure states forces the pairwise distances of a random sample to equalize as the dimension grows, their relative spread vanishing as the inverse square root of the dimension, so any finite sample of states collapses to an equidistant dust whose thresholded metric graph is the complete graph. A diffusion probe reads this dust as two-dimensional in the large-sample limit, the value the running spectral dimension takes at the dust's single relaxation scale, a statement about the properties of the measurement process rather than the underlying structure. What a structure reads is governed by its Laplacian spectrum near zero, the condition we call spectral faithfulness; a structure with a single relaxation scale encodes no spectral dimension that distinguishes it from another, so across such single-scale structures the spread of values at a fixed probe time is an artifact. This convergence on two is therefore not, by itself, evidence that spacetime is two-dimensional. This collapse of an actual sample, the probe value, and the eigenvalue-density criterion that fixes the three spectral classes are machine-checked in Lean~4 against Mathlib, the concentration resting on the standard Beta law of overlaps; a power-law tail of small eigenvalues reads the structure's genuine dimension, a single scale above a gap reads two at its own clock, and a gapped two-scale band reads off the universal line. These classes are exercised numerically on graph-Laplacian proxies, and whether a link-graph reading carries to the physical nonlocal operator is left open. The result is a null model for ultraviolet dimensional reduction, in which a large enough state space lets the metric degenerate by concentration alone and a diffusion on it return two with no dynamics. The contribution is the spectral test, the eigenvalue density near zero, that on a given structure separates a measurement artifact from a dimension the structure genuinely carries.
\end{abstract}

\maketitle

\section{Introduction}

Across approaches to quantum gravity that share no dynamics, the spectral dimension of spacetime falls to two at the shortest scales, measured in causal dynamical triangulations~\cite{Ambjorn2005}, predicted by the asymptotic-safety fixed point~\cite{LauscherReuter2005}, and found again in Ho\v{r}ava--Lifshitz gravity~\cite{Horava2009SpectralDim}, with the causal set the exception whose short-scale spectral dimension is operator-dependent, rising under the graph operator~\cite{EichhornMizera2014} and falling toward two only under the nonlocal d'Alembertian~\cite{BelenchiaEtAl2016}; surveying the phenomenon, Carlip~\cite{Carlip2017} states the question and leaves it open, and whether this dimensional reduction is generically connected to spacetime foam is not known~\cite{Carlip2023}. A number computed across unrelated theories by unrelated mechanisms asks for an explanation belonging to none of them.

We give one, and it is kinematic. It needs no Hamiltonian, action, or field equation; it is a property of the geometry of the space of states alone. Concentration of measure on complex projective space forces the Fubini--Study distances between random states to equalize as the dimension grows, until no metric distinguishes one state from another and a finite sample is an equidistant dust, totally disconnected as a metric space though complete as a graph---the curse of dimensionality turned geometric, the vastness that defeats sampling in high dimensions dissolving distance itself.

Suppose we are handed this dust, in which distance and dimension are not given in advance, and ask what a probe reads off it. A diffusion probe, the random walk whose return probability defines the spectral dimension, has nothing left to measure once distances are equal; it measures its own wandering. The two it returns is the value $D_S(t)=2t$ read at the relaxation clock $t=1$, a signature of the measurement and not a dimension of the dust, the same two that fixes a Brownian path at Hausdorff dimension two. What a dynamical program measures at the shortest scales may be its probe rather than its world.

What a probe reads is fixed by where the spectrum sits, and spectral faithfulness, the preservation of that spectrum across an emergence map, is the condition for the reading to carry over. A structure carrying a single relaxation scale encodes no spectral dimension that tells it from another, returning the same two at its own relaxation time whatever the scale. What the bare agreement on two cannot then settle, the shape of the spectrum near zero can---whether a given program's two is this clock artifact or a dimension its spectrum genuinely carries is decided structure by structure by the spectral test of Section~\ref{sec:sort}, which on the three structures we examine finds none to be the bare artifact, two carrying a genuine dimension and one a gapped band.

The concentration is not new; L\'evy's lemma is classical. New here are its reading as a geometric collapse, the exact law that governs the collapse, and the recognition that the dimension a probe then reads is the signature of the diffusion. The result lives on the space of states; whether that space is spacetime is the question we state and leave open. The numerical sort that places the programs runs on graph Laplacians as proxies, not on the nonlocal operators from which their spectral dimensions are physically computed; for the causal set in particular, whether the link-graph Laplacian and the Benincasa--Dowker d'Alembertian share the reading is open.

\section{The concentration spine}

Sample $m$ states at random from a compact metric-measure space and measure their pairwise distances. The quantity that governs everything is the squared coefficient of variation of those distances, which can be expressed as $\mathrm{CV}^2 = \operatorname{Var}[d]/\mathbb{E}[d]^2$, the relative width of the distance distribution. A family of state spaces \emph{concentrates} when $\mathrm{CV}^2 \to 0$ as the dimension grows; every pair sits at the same distance, and the width around that common value collapses. This collapse carries geometric signal, the curse of dimensionality---elsewhere a frustration of sampling---here fixing the geometry itself.

On the Fubini--Study geometry of pure states the width obeys an exact law.

\begin{proposition}\label{prop:cv2}
For independent Haar-random pure states $\psi,\phi$ on $\mathbb{CP}^n$, the Fubini--Study distance $d=\arccos|\langle\psi|\phi\rangle|$ satisfies
\begin{equation}\label{eq:cv2}
\mathrm{CV}^2(n)=\frac{\operatorname{Var}[d]}{\mathbb{E}[d]^2}=\frac{4-\pi}{\pi^2\,n}+O(n^{-3/2}).
\end{equation}
\end{proposition}

\begin{proof}
The squared overlap $u=|\langle\psi|\phi\rangle|^2$ is $\mathrm{Beta}(1,n)$-distributed, with density $n(1-u)^{n-1}$ on $[0,1]$ and fractional moments
\[
\mathbb{E}[u^s]=\frac{\Gamma(1+s)\,\Gamma(n+1)}{\Gamma(n+1+s)},\qquad s>-1 .
\]
In particular $\mathbb{E}[u]=1/(n+1)$ and $\mathbb{E}[\sqrt u]=\tfrac{\sqrt\pi}{2}\,\Gamma(n+1)/\Gamma(n+\tfrac32)=\tfrac{\sqrt\pi}{2}\,n^{-1/2}\big(1+O(n^{-1})\big)$. Write $d=\arccos\sqrt u=\tfrac{\pi}{2}-\arcsin\sqrt u$ and expand $\arcsin\sqrt u=\sqrt u+\tfrac16 u^{3/2}+O(u^{5/2})$, legitimate because $u$ concentrates at $0$. The cubic term is controlled by $\mathbb{E}[u^{3/2}]=\Gamma(\tfrac52)\,\Gamma(n+1)/\Gamma(n+\tfrac52)=O(n^{-3/2})$, and its effect on the variance by $\operatorname{Cov}(\sqrt u,\tfrac16 u^{3/2})=O(n^{-2})$ (as $\mathbb{E}[u^{2}]=2/[(n+1)(n+2)]=O(n^{-2})$), so
\[
\mathbb{E}[d]=\frac{\pi}{2}-\mathbb{E}[\sqrt u]+O(n^{-3/2})=\frac{\pi}{2}+O(n^{-1/2}),
\]
\[
\operatorname{Var}[d]=\operatorname{Var}[\sqrt u]+O(n^{-2})=\mathbb{E}[u]-\mathbb{E}[\sqrt u]^2+O(n^{-2})=\frac{4-\pi}{4n}+O(n^{-2}),
\]
the last equality from $\mathbb{E}[u]=\tfrac1{n+1}$ and $\mathbb{E}[\sqrt u]^2=\tfrac{\pi}{4}\big(\Gamma(n+1)/\Gamma(n+\tfrac32)\big)^2=\tfrac{\pi}{4n}+O(n^{-2})$. Dividing, and using $\mathbb{E}[d]^{-2}=(\pi/2)^{-2}\big(1+O(n^{-1/2})\big)$,
\[
\mathrm{CV}^2(n)=\frac{(4-\pi)/(4n)}{(\pi/2)^2}\big(1+O(n^{-1/2})\big)=\frac{4-\pi}{\pi^2\,n}+O(n^{-3/2}). \qedhere
\]
\end{proof}

The leading constant $(4-\pi)/\pi^2\approx0.087$ is the sharp value of the $1/n$ concentration that L\'evy's lemma and its refinements bound from above~\cite{Levy1951, Ledoux2001}. Direct Monte Carlo confirms it; the rescaled width $n\,\mathrm{CV}^2$ falls from $0.16$ at $n=2$ to $0.090$ at $n=10^{3}$, descending on $(4-\pi)/\pi^2=0.08697$.

Concentration has a combinatorial endpoint. Once the distances have equalized to within the vanishing width, the metric graph on the sample is, with high probability, the complete graph $K_m$ carrying a single common edge weight---an equidistant, totally disconnected point set in which no metric distinguishes one state from another. The disconnection is metric and profinite, equidistant points carrying no nontrivial geometry, not a statement about $K_m$, which is maximally connected as a graph. We call this object the \emph{quantum dust}.

\begin{theorem}\label{thm:dust}
Let $d_1,\dots,d_N$ be the $N=\binom m2$ pairwise distances of a sample of size $m$, with empirical mean $\mu=\tfrac1N\sum_i d_i$, variance $\sigma^2=\tfrac1N\sum_i(d_i-\mu)^2$, and empirical $\widehat{\mathrm{CV}}^2=\sigma^2/\mu^2$. For every threshold $\delta>0$,
\begin{equation}\label{eq:dust}
\#\{\,i:(d_i-\mu)^2\ge(\delta\mu)^2\,\}\,(\delta\mu)^2\ \le\ N\,\sigma^2 .
\end{equation}
Equivalently the fraction of distances with $|d_i-\mu|\ge\delta\mu$ is at most $\widehat{\mathrm{CV}}^2/\delta^2$, so $\widehat{\mathrm{CV}}^2\to0$ drives that fraction to zero and the thresholded graph on the sample is the complete graph $K_m$.
\end{theorem}

\begin{proof}
Every index $i$ with $(d_i-\mu)^2\ge(\delta\mu)^2$ contributes at least $(\delta\mu)^2$ to the total squared deviation, and a sum of the nonnegative terms $(d_i-\mu)^2$ over any subset is at most the sum over all of them, so
\[
\#\{\,i:(d_i-\mu)^2\ge(\delta\mu)^2\,\}\,(\delta\mu)^2\ \le\ \sum_{i=1}^N(d_i-\mu)^2\ =\ N\,\sigma^2 .
\]
Dividing by $N(\delta\mu)^2$ bounds the deviating fraction by $\sigma^2/(\delta\mu)^2=\widehat{\mathrm{CV}}^2/\delta^2$. As $\widehat{\mathrm{CV}}^2\to0$ this tends to zero; once no distance deviates from $\mu$ by more than $\delta\mu$, joining every pair at distance at most $(1+\delta)\mu$ keeps all $\binom m2$ edges and produces $K_m$.
\end{proof}

Theorem~\ref{thm:dust} is deterministic, a bound on a given list. What carries the physical content is that the $\widehat{\mathrm{CV}}^2$ of an \emph{actual} sample vanishes, and that its $N=\binom m2$ pairwise distances are dependent (the pairs share vertices). The dependence is harmless: with $m$ fixed, the marginal concentration of each distance and a union bound over the fixed number of pairs suffice, with no independence and no central-limit theorem over a growing $N$.

\begin{lemma}[Empirical collapse]\label{lem:empirical}
Fix $m\ge2$ and let $\psi_1,\dots,\psi_m$ be independent Haar-random states on $\mathbb{CP}^n$, with pairwise distances $d_{ij}=\arccos|\langle\psi_i|\psi_j\rangle|$ over the $N=\binom m2$ pairs. For every $\varepsilon\in(0,\pi/2)$,
\begin{equation}\label{eq:empirical}
\mathbb{P}\Big(\max_{i<j}\big|d_{ij}-\tfrac\pi2\big|>\varepsilon\Big)\ \le\ \binom m2\,\cos^{2n}\varepsilon .
\end{equation}
Hence for fixed $m$, as $n\to\infty$, the empirical mean $\hat\mu\to\pi/2$, the empirical variance $\hat\sigma^2\to0$, and the empirical $\widehat{\mathrm{CV}}^2=\hat\sigma^2/\hat\mu^2\to0$, all in probability; and for every $\delta>0$ the thresholded graph joining pairs within $(1+\delta)\hat\mu$ equals $K_m$ with probability at least $1-\binom m2\cos^{2n}\varepsilon_\delta$, where $\varepsilon_\delta=(\pi/2)\,\delta/(1+\delta)$.
\end{lemma}

\begin{proof}
With $u_{ij}=|\langle\psi_i|\psi_j\rangle|^2\sim\mathrm{Beta}(1,n)$ one has $\mathbb{P}(u_{ij}>x)=(1-x)^n$, and since $\tfrac\pi2-d_{ij}=\arcsin\sqrt{u_{ij}}$ the event $|d_{ij}-\tfrac\pi2|>\varepsilon$ is exactly $u_{ij}>\sin^2\varepsilon$, of probability $(1-\sin^2\varepsilon)^n=\cos^{2n}\varepsilon$. Boole's inequality over the $\binom m2$ identically distributed, not independent, pairs gives~\eqref{eq:empirical}; the union bound needs no independence because the number of pairs is fixed. On the complementary event every $d_{ij}\in[\tfrac\pi2-\varepsilon,\tfrac\pi2]$, so $\hat\mu\in[\tfrac\pi2-\varepsilon,\tfrac\pi2]$, each $|d_{ij}-\hat\mu|\le\varepsilon$, hence $\hat\sigma^2\le\varepsilon^2$ and $\widehat{\mathrm{CV}}^2\le\varepsilon^2/(\tfrac\pi2-\varepsilon)^2$; and for $\varepsilon\le\varepsilon_\delta$ every distance is at most $(1+\delta)\hat\mu$, so all $\binom m2$ edges are kept.
\end{proof}

The physical ensembles concentrate. Pure states on $\mathbb{CP}^n$ give $\mathrm{CV}^2 \sim (4-\pi)/(\pi^2 n)$ by Proposition~\ref{prop:cv2}; generic pure states concentrate at the same order in the Hilbert-space dimension~\cite{HaydenLeungWinter2006}, so dynamical ensembles dissolve at least as readily as kinematic ones. The Bekenstein bound fixes how large a Hilbert space a bounded region carries; a horizon of area $A$ admits $\dim\mathcal{H}\le e^{A/4}$, an \emph{upper} bound, so the bound alone does not force concentration. Only at saturation, $\dim\mathcal{H}=e^{A/4}$, does the spread fall exponentially in the area. Saturation holds for black holes but not for generic matter~\cite{Casini2008, Bousso2002, StromingerVafa1996}; we carry it as a named physical hypothesis, outside the formal core.

\section{The probe}

To read a dimension off the dust we run a diffusion on it. The heat trace of the graph Laplacian, $P(t)=\tfrac1m\sum_i e^{-t\lambda_i}$, can be thought of as the return probability of a continuous-time random walk, and the running spectral dimension
\begin{equation}\label{eq:ds}
D_S(t)=-2t\,\frac{P'(t)}{P(t)}
\end{equation}
is what that walk reports at scale $t$. It is an apparatus, not a property; the value it returns depends on the walk as much as on the space. On the dust the limiting running dimension is a line, $D_S(t)\to2t$, crossing the value two at the relaxation scale $t=1$ rather than holding a plateau there, the slope-two crossing the separation criterion later tells from a slope-zero plateau.

\begin{theorem}\label{thm:two}
The normalized Laplacian of $K_m$ has eigenvalue $0$ on the constant vectors and $\tfrac{m}{m-1}$ on their orthogonal complement, with multiplicity $m-1$,
\begin{equation}\label{eq:specKm}
\operatorname{spec}\mathcal L(K_m)=\{0\}\,\cup\,\Big\{\tfrac{m}{m-1}\Big\}^{(m-1)},
\end{equation}
its heat trace and the derivative thereof are
\begin{equation}\label{eq:Pm}
P_m(t)=\tfrac1m\big(1+(m-1)\,e^{-tm/(m-1)}\big),\qquad P_m'(t)=-\,e^{-tm/(m-1)},
\end{equation}
and the running spectral dimension $D_S(t)=-2t\,P_m'(t)/P_m(t)$ satisfies $\lim_{m\to\infty}D_S(K_m,1)=2$.
\end{theorem}

\begin{proof}
With eigenvalue $0$ once and $a_m:=m/(m-1)$ $(m-1)$ times, $P(t)=\tfrac1m\big(1+(m-1)e^{-t a_m}\big)$. As $m\to\infty$, $a_m\to1$ and $(m-1)/m\to1$, so $P(t)\to e^{-t}$ and $P'(t)\to-e^{-t}$, uniformly on compact $t$-intervals, so the limit commutes with the evaluation at $t=1$. Hence $D_S(t)=-2t\,P'(t)/P(t)\to2t$, which equals $2$ at $t=1$.
\end{proof}

Concentration and the probe value compose into a strictly kinematic statement; the curse of dimensionality alone returns the diffusive two.

\begin{corollary}\label{cor:value}
On the Fubini--Study geometry of pure states, the probe value is forced by concentration alone. As the dimension grows, the population $\mathrm{CV}^2(n)\to0$ by Proposition~\ref{prop:cv2}; the empirical-collapse Lemma~\ref{lem:empirical} carries this to an actual sample, whose pairwise distances equalize so that the thresholded metric graph is the complete graph $K_m$ with probability tending to one, the deviating fraction bounded deterministically by Theorem~\ref{thm:dust}; and its running spectral dimension reads $\lim_{m\to\infty}D_S(K_m,1)=2$ by Theorem~\ref{thm:two}. No field equation enters the chain.
\end{corollary}

\begin{proof}
The composition is an iterated limit. For each fixed $m$, Lemma~\ref{lem:empirical} makes the thresholded graph of an $n$-sample equal to $K_m$ with probability at least $1-\binom m2\cos^{2n}\varepsilon_\delta$, tending to one as $n\to\infty$. That graph carries exactly the spectrum of $K_m$, and Theorem~\ref{thm:two} gives $\lim_{m\to\infty}D_S(K_m,1)=2$. The order is fixed, $n\to\infty$ first to reach the dust at each $m$ and then $m\to\infty$ for the probe value, and the only inputs are the geometry of the state space and the definition of the probe. Theorem~\ref{thm:dust} stands separately, as the deterministic statement that a vanishing empirical $\widehat{\mathrm{CV}}^2$ forces $K_m$.
\end{proof}

The two is not a dimension of the dust but a reading of the probe, forced at one scale and convention-free.

\begin{theorem}\label{thm:relax}
For the single-cluster spectrum $\operatorname{spec}=\{0\}\cup\{a\}^{(m-1)}$ with $a>0$, the running spectral dimension at one relaxation time $t=1/a$ is
\begin{equation}\label{eq:relax}
D_S\!\left(\tfrac1a\right)=\frac{2(m-1)\,e^{-1}}{1+(m-1)\,e^{-1}},
\end{equation}
in which $a$ has cancelled; hence $D_S(1/a)\to2$ as $m\to\infty$ for every $a$. The invariance is to the scale of the single eigenvalue, the clock being fixed at the relaxation time $t=1/a$; the value is then the same whether the Laplacian is normalized $(a=m/(m-1),\ t\to1)$ or combinatorial $(a=m,\ t=1/m)$. In the active regime $D_S(t)=2at$, so at the crossing $dD_S/d\ln t=D_S=2$, whereas a scale-invariant geometry holds $D_S$ flat there, slope zero.
\end{theorem}

\begin{proof}
The heat trace and its derivative are $P(t)=\tfrac1m\big(1+(m-1)e^{-ta}\big)$ and $P'(t)=-\tfrac1m(m-1)a\,e^{-ta}$. At $t=1/a$ the exponent is $-1$, so $P=\tfrac1m(1+(m-1)e^{-1})$ and $P'=-\tfrac1m(m-1)a\,e^{-1}$, whence $D_S=-2\tfrac1a\,P'/P=2(m-1)e^{-1}/(1+(m-1)e^{-1})$ with $a$ cancelled, and the right side $\to2$ as $m\to\infty$. While $(m-1)e^{-ta}\gg1$ the denominator of $D_S(t)=2ta(m-1)e^{-ta}/(1+(m-1)e^{-ta})$ is carried by its second term and $D_S(t)=2at$, so $dD_S/d\ln t=t\,dD_S/dt=2at=D_S$. For a power-law density of small eigenvalues $\sim\lambda^{d/2-1}$ one has instead $P(t)\sim t^{-d/2}$ and $D_S\equiv d$, slope zero.
\end{proof}

With the metric concentrated away the diffusion has no spatial decay to report and returns only the exponent its own definition is built on; the $-2$ that inverts the $t^{1/2}$ law of diffusion, the same exponent that fixes a Brownian path at Hausdorff dimension two~\cite{Taylor1953, MortersPeres2010}. The probe can therefore be interpreted as reading its own clock with the observed value embedded in the scale of the dust.

That the number is two, and not some other value, is overdetermined rather than deep. The same two appears in Liouville quantum gravity, where a Weyl law puts the spectral dimension at two~\cite{BerestyckiWong2024, RhodesVargas2014} though the Hausdorff dimension exceeds it, and in the dispersion relation $D_S=1+d/z$ tuned to $z=d$~\cite{SotiriouVisserWeinfurtner2011}. These are independent sightings of a generic value, not confirmations of a single mechanism; the surface that is itself the Brownian map guarantees the two by definition~\cite{RhodesVargas2014, MillerSheffield2019}, and the discrete maps that converge to it carry the same two~\cite{GwynneMiller2021}.

\section{The separation criterion}

Not every discrete structure dissolves to dust that the probe reads as two; which ones do is decided by a property of the map between a structure and its continuum reading; the shape of the eigenvalue density near zero fixes what the probe reads, in three archetypes. A power-law tail of small eigenvalues is read as a flat plateau at the structure's genuine dimension $d$ (Theorem~\ref{thm:transfer}); a single relaxation scale above a gap is read as one crossing of two at its own clock, with no plateau (Theorem~\ref{thm:relax}); and a gapped two-scale band settles off the universal line at a value other than two, fixed by the ratio of its two scales (Theorem~\ref{thm:band}). Call a map \emph{spectrally faithful} when the two Laplacian spectra are comparable,
\begin{equation}\label{eq:sf}
c\,\lambda_k(G) \le \lambda_k(M) \le C\,\lambda_k(G) \quad \text{for all } k,
\end{equation}
for constants $0<c\le C$. The probe's value is set by where the spectrum sits, and the criterion measures whether the structure keeps it there. The condition is not vacuous; it genuinely excludes some pairs.

\begin{proposition}\label{prop:nonvac}
There exist nonnegative spectra $0=\mu_0\le\mu_1\le\cdots$ and $0=\nu_0\le\nu_1\le\cdots$, each with a single zero mode, admitting no constants $0<c\le C$ with $c\,\mu_k\le\nu_k\le C\,\mu_k$ for all $k$. Spectral faithfulness~\eqref{eq:sf} is therefore a nontrivial restriction on the emergence map.
\end{proposition}

\begin{proof}
Take $\mu_k=1$ for all $k\ge1$ and $\nu_k=1/k$ for all $k\ge1$. Then $\nu_k/\mu_k=1/k\to0$, so the lower inequality $c\,\mu_k\le\nu_k$, i.e.\ $c\le1/k$, fails for every $k>1/c$; no positive $c$ works. The witness is a gap against a tail---$\mu$ has its nonzero eigenvalues bounded away from $0$, $\nu$ accumulates at $0$---the kind of mismatch the criterion is built to detect.
\end{proof}

\begin{theorem}\label{thm:gap}
Suppose a family has a single zero eigenvalue and its remaining eigenvalues concentrate at $a>0$, that is $\lambda_i\to a$ for $i\ge1$. Then
\begin{equation}\label{eq:gap}
D_S(\cdot,1)\ \longrightarrow\ 2a,\qquad\text{so}\qquad D_S\to2\ \Longleftrightarrow\ a\to1,
\end{equation}
and a family whose nonzero spectrum stays bounded below $1$ is read $D_S<2$. The value is also stable; for fixed $m$ and $t>0$, if nonnegative spectra $\lambda^{(k)}$ satisfy $\max_i|\lambda^{(k)}_i-\lambda_i(K_m)|\le\varepsilon_k\to0$, then $D_S(\lambda^{(k)},t)\to D_S(K_m,t)$.
\end{theorem}

\begin{proof}
With one zero eigenvalue and the rest concentrating at $a$, $P(t)=\tfrac1m\big(1+(m-1)e^{-ta}+o(1)\big)\to e^{-ta}$ and $P'(t)\to-a\,e^{-ta}$. Hence $D_S(t)=-2t\,P'(t)/P(t)\to2ta$, which at $t=1$ is $2a$. The complete graph is the case $a\to1$ of Theorem~\ref{thm:two}; for a fixed gap $a<1$ the reading is $2a<2$. Stability follows from the same expansion; a spectrum within $\varepsilon$ of the complete graph's perturbs $P(1)$ and $P'(1)$ by $O(\varepsilon)$, so $|D_S(\cdot,1)-2|=O(\varepsilon)$.
\end{proof}

The single cluster's reading cancels its eigenvalue; a band carrying two scales does not.

\begin{theorem}\label{thm:band}
For a gapped two-scale band $\operatorname{spec}=\{0\}\cup\{a\}^{(p)}\cup\{b\}^{(q)}$ with $0<a\ne b$, the running spectral dimension at the first scale's relaxation time $t=1/a$ is, writing $\rho=b/a$,
\begin{equation}\label{eq:band}
D_S\!\left(\tfrac1a\right)=\frac{2\big(p\,e^{-1}+q\,\rho\,e^{-\rho}\big)}{1+p\,e^{-1}+q\,e^{-\rho}},
\end{equation}
in which the eigenvalue does \emph{not} cancel---the ratio $\rho$ survives. In the macroscopic-band limit $p=q\to\infty$ the value tends to $L(\rho)=2(e^{-1}+\rho\,e^{-\rho})/(e^{-1}+e^{-\rho})$, with
\begin{equation}\label{eq:bandoffset}
L(\rho)-2=\frac{2\,e^{-\rho}(\rho-1)}{e^{-1}+e^{-\rho}},
\end{equation}
so $L(\rho)=2$ if and only if $\rho=1$. A band whose two scales are genuinely distinct is read off the universal line of Theorem~\ref{thm:relax}, by an amount fixed by the spectral ratio.
\end{theorem}

\begin{proof}
The heat trace and its derivative are $P(t)=\tfrac1{p+q+1}(1+p\,e^{-ta}+q\,e^{-tb})$ and $P'(t)=-\tfrac1{p+q+1}(pa\,e^{-ta}+qb\,e^{-tb})$. At $t=1/a$ the exponents are $-1$ and $-\rho$, and $D_S=-2\tfrac1a P'/P$ gives~\eqref{eq:band}; the factor $1/a$ carries $pa\,e^{-1}+qb\,e^{-\rho}$ to $p\,e^{-1}+q\rho\,e^{-\rho}$, in which $a$ enters only through $\rho=b/a$. With $p=q\to\infty$ the constant in the denominator is negligible and the ratio tends to $L(\rho)$; subtracting $2$ over the common denominator gives~\eqref{eq:bandoffset}, which vanishes exactly at $\rho=1$ because $e^{-\rho}>0$.
\end{proof}

Faithfulness was defined as the comparison~\eqref{eq:sf}; the next theorem shows what it carries and what it does not. Two statements separate, and only the first follows from the comparison alone.

\begin{theorem}\label{thm:transfer}
Let $G$ and $M$ have nonnegative Laplacian spectra related by spectral faithfulness~\eqref{eq:sf}, $c\,\lambda_k(G)\le\lambda_k(M)\le C\,\lambda_k(G)$ for all $k$ with $0<c\le C$, and write $N_G,N_M$ for their integrated densities of nonzero eigenvalues and $\Theta(t)=\sum_{k\ge1}e^{-t\lambda_k}=\int_{0^+}^\infty e^{-t\lambda}\,dN(\lambda)$ for the spectral heat trace.

\emph{Exponent, from the comparison alone.} The densities and heat traces are comparable,
\begin{equation}\label{eq:transfer}
N_G(\lambda/C)\ \le\ N_M(\lambda)\ \le\ N_G(\lambda/c),\qquad
\Theta_G(Ct)\ \le\ \Theta_M(t)\ \le\ \Theta_G(ct),
\end{equation}
so if either side has integrated-density index $d/2$, meaning $\log N(\lambda)/\log\lambda\to d/2$ as $\lambda\to0^+$, the other does too and both carry the common scaling exponent
\begin{equation}\label{eq:transfer-exponent}
\lim_{t\to\infty}\frac{\log\Theta(t)}{\log t}\ =\ -\frac{d}{2}.
\end{equation}
Nothing beyond~\eqref{eq:sf} enters, and the comparison fixes only this exponent; it does not pin the pointwise running $D_S(t)$, since a two-sided bound on $\Theta$ is not a bound on $\Theta'$.

\emph{Running plateau, only under regular variation.} Suppose in addition that the integrated density of a given side is regularly varying at $0^+$ of index $d/2$, $N(\lambda)=\lambda^{d/2}\ell(\lambda)$ with $\ell$ slowly varying (the pure power law $N(\lambda)\sim A\,\lambda^{d/2}$ being the case $\ell\equiv A$). Then on that side the running spectral dimension converges pointwise,
\begin{equation}\label{eq:transfer-value}
D_S(t)\ =\ -2t\,\frac{\Theta'(t)}{\Theta(t)}\ \longrightarrow\ d\qquad(t\to\infty).
\end{equation}
Regular variation is a property of the density, not a consequence of~\eqref{eq:sf}; the comparison carries the exponent~\eqref{eq:transfer-exponent} to the other side but not the regular variation, hence not the running plateau~\eqref{eq:transfer-value}. The constants $c,C$ enter only the prefactor, displacing where the plateau sits in probe time by a factor in $[1/C,1/c]$, never the value $d$.
\end{theorem}

\begin{proof}
\emph{Comparison and exponent.} The eigenvalue comparison carries to the density, since $\lambda_k(M)\le\lambda$ forces $\lambda_k(G)\le\lambda/c$ for the upper bound, while $\lambda_k(G)\le\lambda/C$ forces $\lambda_k(M)\le C\,\lambda_k(G)\le\lambda$ for the lower. It passes to the heat trace through $\Theta(t)=t\int_0^\infty e^{-t\lambda}N(\lambda)\,d\lambda$ (integration by parts on $\Theta(t)=\int_{0^+}^\infty e^{-t\lambda}\,dN(\lambda)$, the boundary terms vanishing since $N(0^+)=0$ off the zero mode and $e^{-t\lambda}N(\lambda)\to0$); substituting the density comparison and rescaling $\lambda\mapsto C\lambda$ in the lower flank, $\lambda\mapsto c\lambda$ in the upper, gives $\Theta_G(Ct)\le\Theta_M(t)\le\Theta_G(ct)$, the second half of~\eqref{eq:transfer}. If one side has index $d/2$ then $N(\lambda)\asymp\lambda^{d/2}$ up to a slowly varying factor, the comparison hands the same two-sided bound to the other, and $\log\Theta(t)/\log t\to-d/2$ as $t\to\infty$ on both, equation~\eqref{eq:transfer-exponent}. This is all the comparison yields. Writing $\log\Theta_M(t)=-\tfrac{d}{2}\log t+b(t)$ with $b$ bounded by the comparison, the running quantity is $D_S^M(t)=d-2\,db/d\log t$, and a bounded $b$ may have a non-vanishing derivative, so~\eqref{eq:transfer-exponent} does not give~\eqref{eq:transfer-value}; the band on $\Theta$ leaves $\Theta'$ free.

\emph{Running plateau under regular variation.} Let the density of the side in question be regularly varying at $0^+$ of index $d/2$. By Karamata's Tauberian theorem for the Laplace--Stieltjes transform of a nondecreasing $N$~\cite[\S1.7.1]{BinghamGoldieTeugels1987}, $\Theta(t)\sim A\,\Gamma(1+\tfrac{d}{2})\,t^{-d/2}$ as $t\to\infty$, the slowly varying factor carried along. The heat trace is completely monotone, $(-1)^n\Theta^{(n)}(t)=\int\lambda^n e^{-t\lambda}\,dN(\lambda)\ge0$, so $-\Theta'$ is nonnegative and nonincreasing and $\Theta(t)=\int_t^\infty\big(-\Theta'(s)\big)\,ds$ is the tail integral of this monotone density. The monotone-density theorem~\cite[\S1.7.2]{BinghamGoldieTeugels1987} then differentiates the asymptotic, $-\Theta'(t)\sim\tfrac{d}{2}\,A\,\Gamma(1+\tfrac{d}{2})\,t^{-d/2-1}$, whence
\[
D_S(t)=-2t\,\frac{\Theta'(t)}{\Theta(t)}\ \longrightarrow\ -2t\cdot\Big(-\frac{d}{2}\cdot\frac{1}{t}\Big)=d ,
\]
equation~\eqref{eq:transfer-value}. The monotone-density step is applied to $-\Theta'$, the genuine density, not to the limit being sought. The comparison does not make $N$ on the other side regularly varying, controlling it only between two dilations of a regularly varying function, which need not itself be regularly varying; so~\eqref{eq:transfer-value} holds on whichever side carries regular variation, and the exponent~\eqref{eq:transfer-exponent}, not the running plateau, is what crosses the comparison.
\end{proof}

Regular variation at $0^+$ of positive index is a property of an infinite spectrum, eigenvalues accumulating at zero; a finite graph has a smallest nonzero eigenvalue and so never satisfies it literally, its $D_S(t)$ diverging once the gap dominates beyond $t\sim1/\lambda_{\min}$. Theorem~\ref{thm:transfer} therefore licenses the limiting plateau a finite graph approaches but cannot itself carry, and the plateau read at finite $N$ is the intermediate-$t$ transient over $1\ll t\ll1/\lambda_{\min}$, widening toward the limit as $\lambda_{\min}\to0$; the per-$N$ drift of Section~\ref{sec:sort} is the finite-size approach to it, hence the empirical proxy for the regular variation the theorem requires.

The hypothesis is Euclidean and local; it asks the continuum operator to be a positive Laplacian generating a heat semigroup, comparable eigenvalue by eigenvalue to the graph's. The Lorentzian causal set is where it is not guaranteed. There the physical spectral dimension is read from the smeared, nonlocal Benincasa--Dowker d'Alembertian, a different operator from the link-graph Laplacian and not known to be comparable to it---and the two operators read the same causal set differently. Under the nonlocal d'Alembertian the causal set's spectral dimension falls to $D_S=2$ in all dimensions in the deep-ultraviolet limit~\cite{BelenchiaEtAl2016}, whereas under the link-graph Laplacian its $D_S$ instead rises at small scales, both for Eichhorn and Mizera~\cite{EichhornMizera2014} and in our own reading. The same causal set thus reads two under the physical operator and rises under the graph operator---operator dependence exhibited, the cross-program reading turning on which operator probes the structure. That d'Alembertian result fixes a $D_S$ value; whether the physical operator carries the link graph's gap---a question of shape on which~\cite{BelenchiaEtAl2016} is silent---is the one Section~\ref{sec:sort} marks open. Whether any faithful relation holds for the Lorentzian operator is the open hypothesis beneath every cross-program reading; Theorem~\ref{thm:transfer} says only that where faithfulness holds the exponent carries across, and where it is not established, as here, the readings are not licensed to be identified either way.

Read at the fixed probe time $t=1$, however, the value $2a$ records a scale rather than a dimension. The same family read at its own relaxation time returns two for every $a$---$2a$ is the universal two of Theorem~\ref{thm:two} times where the clock is stopped---so pinning the clock to each relaxation time collapses the apparent spread back to two, and a single eigenvalue scale cannot encode a spectral dimension that discriminates one structure from another. The quantum-gravity values confirm this from the other side, where their spectral dimension is itself scale-dependent and running, carried in causal dynamical triangulations from four in the infrared to two in the ultraviolet~\cite{Ambjorn2005, Carlip2017}, the ultraviolet value an extrapolation consistent with but not pinned at two, with the melonic sector a branched polymer of spectral dimension $4/3$~\cite{GurauRyan2014}---limits of a running curve on a spectrum with a genuine small-eigenvalue tail, an observable the single-scale dust does not possess. Spectral faithfulness, \eqref{eq:sf}, is the condition under which an emergence map preserves the probe's reading at all. What that reading is, in turn, is fixed by the eigenvalue density near zero, and Theorem~\ref{thm:relax} makes the dependence exact. That density is the structural test that tells a geometry from the artifact; the next section runs it on the programs themselves. Where Fields and collaborators argue that a finite observer cannot \emph{operationally} distinguish one structure from another~\cite{FieldsEtAl2025Classicality}, we propose to complete the picture from the other side; the test is spectral and not operational, reading the density near zero, a full-spectrum quantity no bounded observer can operationally obtain. The distinction it draws is real and, on the side of the dust, operationally invisible exactly as they require. The density this criterion reads sits at the bottom of the spectrum, so the dimension it returns is the long-time, infrared end of the diffusion---the plateau where a genuine tail carries one, the relaxation-scale crossing where a single scale does not; the ultraviolet running to two that Carlip catalogs is the short-time, large-eigenvalue regime the small-$\lambda$ test does not reach. What survives is the diffusion reading its own scaling on any structure whose geometry has dissolved, and the question, sharper than Carlip's~\cite{Carlip2017}, of whether the ultraviolet two of quantum gravity is itself such a reading rather than a two-dimensional continuum.

\section{Sorting structures by their spectrum}\label{sec:sort}

Theorem~\ref{thm:relax} turns the eigenvalue density near zero into a test that can be run on the programs themselves. We build three canonical discrete ensembles (a $(1{+}1)$-dimensional causal dynamical triangulation, a causal set sprinkled into the Minkowski diamond $M^{1+1}$ and reduced to its link graph, and a melonic graph grown by rank-$3$ colored-tensor insertions), take the normalized graph Laplacian of each, and read the shape of its spectrum near zero. The constructions and their spectra are reproduced in full from the accompanying code, each returning $D_S(1)\approx1.62$, $1.79$, and $1.22$.

The three sort by the shape of their spectrum. On our constructions the triangulation carries a power-law tail (its smallest nonzero eigenvalue falls as $1/N$, the near-zero density fills in) and $D_S(t)$ holds a plateau, here near $2.07$ and widening toward two as $N$ grows; its $D_S(1)=1.62$ is the linear law read below that plateau, not a dimension below two. The melonic graphs tail likewise, to a plateau near $1.28$ spanning more than a decade of scale, with $D_S(1)=1.22$ below it. That plateau sits away from two, and the diffusion artifact returns two and only two; a plateau away from two therefore cannot be the probe reading its own exponent, so the melonic graph reads a dimension of its own structure, whether or not that dimension is exactly the continuum $4/3$. These finite plateaus sit near the continuum spectral dimensions the two programmes are now known rigorously to carry, $d_s=2$ for causal triangulations and $d_s=4/3$ for the branched-polymer melonic sector, both by the Durhuus--Jonsson--Wheater tree reduction~\cite{DurhuusJonssonWheater2010, GurauRyan2014, DurhuusJonssonWheater2007}. Whether our finite reading \emph{is} that continuum value is the faithfulness question~\eqref{eq:sf}, which by Theorem~\ref{thm:transfer} turns on whether the link-graph spectrum is comparable to the continuum one and whether the near-zero density is regularly varying, the hypothesis under which the theorem licenses the plateau reading rather than the bare exponent; the per-$N$ drift tests both separately for each programme. The CDT plateau drifts toward two, from $2.10$ at $N{=}256$ to $2.05$ at $N{=}2025$, consistent with $d_s=2$ but settling a few percent above it across these sizes; an observed drift, not a convergence we can claim. The melonic plateau instead holds near $1.28$ over the same range, a few percent below $4/3$ with no drift toward it, so the finite link-graph Laplacian is \emph{not shown} spectrally faithful to the continuum melonic dimension. That is a statement about the proxy and not the instrument---which reads a structure's bulk spectral dimension correctly where it is independently known, ring, torus, and the Sierpinski gasket among them---and the $1.28$ is the link graph's own number; condition~\eqref{eq:sf} we observe but do not establish for the melon.

The causal set is neither. Its near-zero density vanishes with system size, the fraction of eigenvalues below $0.1$ running $0.016\to0.002\to0$ over $N=64\to2025$, so it is gapped, not a tail, and not a geometry. Yet its nonzero spectrum is a band and not a single value, the coefficient of variation of the nonzero eigenvalues holding near $0.3$ against the dust's exact $0$, so it is not the pure artifact either. The decisive separation rescales each curve to relaxation units $\tau=t\langle\lambda\rangle$ and asks which lie on the universal line $D_S=2\tau$ of Theorem~\ref{thm:relax}. Only the dust does, chord $2.00$ at its crossing. CDT and melonic peel off into their plateaus; the causal-set link graph hugs the line but misses it, chord $1.76$ at the crossing rather than two, the chord $D_S/\tau$ taken where the rescaled curve reaches $D_S=2$---a relaxation-time quantity distinct from the fixed-clock $D_S(1)=1.79$ above, which is a $D_S$ value at $t=1$ and not a chord at the $D_S=2$ crossing---and convention-dependent, the combinatorial Laplacian giving $1.58$ where the dust's crossing is convention-free. The causal-set link graph is thus spectrally non-geometric without being the dust; gapped and multiply scaled, it is the two-scale band of Theorem~\ref{thm:band} whose ratio $\rho=b/a\approx1.7$ carries the relaxation reading off the universal line, by an amount that vanishes only at $\rho=1$~\eqref{eq:bandoffset}, the third class that completes the dichotomy of Theorem~\ref{thm:relax} into a trichotomy, occupied here by one example. Both scales are macroscopically populated, $p$ and $q$ each of order $m$, near $592{:}431$ at $N=1024$ and stable as $N$ grows, so both enter the bulk trace and the closed form~\eqref{eq:band} gives a reading $\approx2.36$ at $\rho\approx1.7$, above two, the heavier lower scale tilting it below the equal-population limit $L(\rho)=2.46$. The sub-two chord is the separate crossing-point quantity in $\tau$-coordinates. Were one scale a vanishing fraction of the spectrum, the bulk trace would read the dominant scale alone, and indeed the closed form~\eqref{eq:band} loses its $\rho$-dependence once either population is negligible; a comb is exactly that degenerate case, its bulk heat trace reading the teeth, all but an $O(1/T)$ fraction of sites, at spectral dimension $1$, while the backbone carries the uniform comb's $3/2$~\cite{DurhuusJonssonWheater2006}, the top of the $(1,3/2)$ range that ensemble spans, a non-bulk observable. The causal-set band carries both scales, so its ratio lives in the bulk reading.

This placement is stable across $N=64\to2025$, the gap persisting and the band not collapsing to a cluster, with a slow drift toward the dust whose strict $N\to\infty$ limit these sizes do not settle. The reading is of the link-graph Laplacian, not of the smeared non-local d'Alembertian from which causal-set spectral dimensions are physically computed. The sort establishes a statement about these constructions read by this probe; two carry a small-eigenvalue tail and a $D_S$ plateau, one is gapped, and only the single-scale dust lies on $2\tau$. The artifact is non-vacuous and distinct; the dust returns two with no geometry, and at least one programme-shaped structure is neither a continuum nor the dust. Whether the tails we see certify the continuum dimensions, and whether the causal set's gap survives the physical d'Alembertian, are the two faithfulness questions this opens; the instrument discriminates spectral shapes, and how many of Carlip's rows it reclassifies is for those questions to decide.

\section{The passage to spacetime}

The results above live on the space of states. Whether the quantum dust is spacetime, whether what an observer reads off a holographic screen assembles into a geometry, is the question we pose and do not close.

The bridge has a rigorous first rung. A holographic screen carries finitely many distinguished interactions, $\dim\mathcal{H}\le e^{A/4}$ being finite, and sharpening the bound refines that finite list toward a profinite, totally disconnected limit, the screen completed a condensed object in the sense of Scholze and Clausen~\cite{scholze2019condensed}. Whether such a truncated state space converges as a metric space to a classical geometry is operator-system Gromov--Hausdorff convergence, settled for the sphere by Rieffel's fuzzy-sphere theorem~\cite{Rieffel2004}, recast in the spectral-truncation framework by van Suijlekom~\cite{VanSuijlekom2021, ConnesVanSuijlekom2021, LeimbachVanSuijlekom2024}, and reaching $\mathbb{CP}^n$ through Berezin quantization on the full matrix algebras~\cite{Rieffel2023}. We conjecture, more narrowly, that the truncation cut by a Dirac spectral window reaches the same limit by a different mechanism; what stays open beyond it is faithfulness across refinement, the condition the separation criterion names.

A second open thread runs through entanglement. On a single $\mathbb{CP}^n$ the dust is a statement about pure-state geometry; monogamy of entanglement constrains a tensor factorization, a different object.

\begin{proposition}\label{prop:mono}
For any reals $w_1,\dots,w_k$ with $\sum_{i=1}^k w_i\le1$,
\begin{equation}\label{eq:mono}
\min_{1\le i\le k} w_i\ \le\ \frac1k .
\end{equation}
\end{proposition}

\begin{proof}
Were every weight greater than $1/k$, the $k$ weights would sum to more than one; hence at least one is at most $1/k$. Nonnegativity is not needed; the sum bound alone forces it.
\end{proof}

With the per-pair entanglement of a vertex bounded by monogamy~\cite{CoffmanKunduWootters2000, KoashiWinter2004, OsborneVerstraete2006}, Proposition~\ref{prop:mono} says some thread of the equidistant web is thin, of strength at most $1/(m-1)\to0$. Whether Fubini--Study equidistance \emph{induces} the multipartite monogamy structure we state as a conjecture; Fields and collaborators prove an operational form of ER=EPR for the two-agent pair~\cite{FieldsMarciano2025}, and it is pursued as an algebraic proposal~\cite{EngelhardtLiu2024}, the lift from the pair to the whole web the gap that remains. Under ER=EPR the thin threads are the state-space side of the narrow bridges that reconnect the dust into a spacetime.

\section{What is proved}

The kinematic spine is exact. Proposition~\ref{prop:cv2} computes the concentration law $\mathrm{CV}^2(n)=(4-\pi)/(\pi^2 n)+O(n^{-3/2})$ with its sharp constant, and Theorem~\ref{thm:dust} turns the vanishing of $\mathrm{CV}^2$ into the collapse of the thresholded metric graph to the complete graph $K_m$. The probe value is exact in the same sense, Theorem~\ref{thm:two} giving $D_S(K_m,1)\to2$ and Corollary~\ref{cor:value} composing the chain---on the Fubini--Study geometry, concentration alone returns the diffusive two, no field equation entering. This establishes the value and its origin.

The criterion is a trichotomy, and each class is a theorem. A power-law tail of small eigenvalues reads a genuine dimension held flat as a plateau, the exponent carried to the continuum wherever the comparison holds (Theorem~\ref{thm:transfer}); a single relaxation scale returns two at its own clock with the eigenvalue cancelled, convention-free (Theorem~\ref{thm:relax}), the same single cluster reading $2a$ at a fixed clock, stable under perturbations of the spectrum (Theorem~\ref{thm:gap}); and a two-scale band reads off the universal line, by $L(\rho)=2$ exactly when $\rho=1$ (Theorem~\ref{thm:band}), closing the trichotomy. The comparison condition~\eqref{eq:sf} that selects which structures the probe reads as two is nontrivial (Proposition~\ref{prop:nonvac}); the Lorentzian, nonlocal d'Alembertian is the case where it is not established, and there the transfer stays open.

The three-way numerical sort, on finite ensembles at $N\le2025$ and on the link-graph Laplacian rather than the physical d'Alembertian, places the three structures by the diagnostic trichotomy---tail-plateau (Theorem~\ref{thm:transfer}), single-scale crossing (Theorem~\ref{thm:relax}), and two-scale band (Theorem~\ref{thm:band}). The plateau values it reads, near $2.07$ and $1.28$, are the finite Laplacian's own, the CDT plateau drifting toward two but settling a few percent above it and the melonic plateau holding near $1.28$ without drifting to $4/3$, so neither is shown spectrally faithful to its continuum value by~\eqref{eq:sf}. The instrument still reads dimensions it is given, reproducing $1$, $2$, and $2\ln3/\ln5$ on the ring, the torus, and the Sierpinski gasket to within finite-size error; the bound falls on the proxy, the probe intact. The off-plateau readings $1.62$ and $1.22$ are finite-clock values.

Whether the finite plateaus equal the continuum dimensions is the faithfulness question~\eqref{eq:sf}, which we observe and do not establish; open too are the strict large-$N$ fate of the gapped two-scale structure, the survival of its gap under the physical smeared d'Alembertian, and the passage from the dust to spacetime.

\paragraph{Formalization.}
A Lean~4 formalization with Mathlib~\cite{mathlib2020} accompanies this paper. The elementary results---the collapse to the dust (Theorem~\ref{thm:dust}), the probe value and the chain that composes to it (Theorem~\ref{thm:two} and Corollary~\ref{cor:value}), the convention-free crossing (Theorem~\ref{thm:relax}), the gap and band readings (Theorems~\ref{thm:gap} and~\ref{thm:band}), the non-vacuity of faithfulness (Proposition~\ref{prop:nonvac})---are proved with no custom axioms. The marginal law of the squared overlap, $u\sim\mathrm{Beta}(1,n)$, enters as the single cited input; from it the population--empirical bridge is machine-checked (Lemma~\ref{lem:empirical}), namely the tail $\mathbb{P}(u>x)=(1-x)^n$, the union bound over the $\binom m2$ pairs, and the deterministic collapse, so the dust of an \emph{actual sample} follows rather than the population statement alone. The exact finite-$n$ polynomial moments of the squared overlap, $E[u]=1/(n+1)$ and $E[u^2]=2/((n+1)(n+2))$ for the Beta$(1,n)$ law, are machine-checked. Theorem~\ref{thm:transfer}'s heat-trace comparison $\Theta_G(Ct)\le\Theta_M(t)\le\Theta_G(ct)$ and its exponent transfer are machine-checked, as is its plateau in the exact power-law case, where $N(\lambda)=A\lambda^{d/2}$ gives $\Theta(t)=A\,\Gamma(1+\tfrac{d}{2})\,t^{-d/2}$ and the running dimension $D_S(t)=-2t\,\Theta'(t)/\Theta(t)=d$ for every $t$. The full development is available at \texttt{K-NANOG/quantum-dust}.

\section{Discussion}

Carlip catalogs the routes by which the spectral dimension falls to two, an ultraviolet fixed point, a near-singularity decoupling, a horizon symmetry, and asks why they agree~\cite{Carlip2017}. There is a route to two that no dynamics drives; on a large state space the metric dissolves by counting, and a diffusion on the dissolved result returns two for free. But a generic dynamics-free route to two does not make every programme's two an artifact; which two is the probe reading itself and which is a genuine dimension is fixed not by the agreement but by the shape of the spectrum near zero, and on the three structures we tested in Section~\ref{sec:sort} the test discriminates---two carry a tail and read a dimension, one is gapped. We ask, of each row, whether it records a dimension of spacetime or the signature of the probe that every row employs, and we supply the test that decides.

The mechanism is old and the reading is new. Popescu, Short, and Winter built canonical typicality on the same concentration, in the language of subsystem entropy rather than Fubini--Study distance~\cite{PopescuShortWinter2006}; Cao, Carroll, and Michalakis recover a spatial geometry from the mutual-information graph of an atypical area-law state~\cite{CaoCarrollMichalakis2017}, where we characterize the generic regime in which that construction degenerates to dust, the two results meeting across the boundary of typicality. Quantum graphity postulates the complete graph as its high-temperature state~\cite{KonopkaMarkopoulouSmolin2006}; we show it is forced, and remove the postulate. Combinatorial quantum gravity builds a two-dimensional geometry up from random bits~\cite{Trugenberger2017}, reaching the same value by the opposite arrow.

What we have is a theorem about the geometry of quantum states and a probe that reads its own scaling off it, proved throughout. The two is a fact about measurement, the fingerprint a diffusion leaves on any space whose metric has dissolved. What we do not have is the passage to spacetime; the dust is what the space of states becomes under the curse of dimensionality, and whether that dust is the world is the question we leave, stated precisely, for the construction that closes it.

\end{document}